\documentclass[10pt,journal,twocolumn,twoside]{IEEEtran}
\usepackage{amsfonts,amsbsy,bm,amsmath}
\usepackage[nolist]{acronym}
\usepackage{mathrsfs} 
\usepackage{graphicx,cite,amssymb,mathtools,amsthm}
\usepackage{stfloats}
\usepackage{xcolor}
\usepackage{tikz}
\usepackage{pgfplots}
\usepackage{pgf}
\mathtoolsset{showonlyrefs}
\usepackage{bbm}
\usepackage{mathabx}
\usepackage{units}
\usetikzlibrary{arrows.meta,positioning,fit,calc,backgrounds}
\allowdisplaybreaks

\usepackage{subcaption}
\definecolor{MustafaColor}{RGB}{0,92,175}
\definecolor{AhmedColor}{RGB}{0,140,110}
\definecolor{AvijitColor}{RGB}{180,55,120}
\definecolor{HomaColor}{RGB}{220,120,30}
\definecolor{AssignColor}{RGB}{180,30,30}

\newtheorem{lemma}{Lemma}
\newtheorem{theorem}{Theorem}

 \acrodef{5G-NR}[5G-NR]{5G new radio}
\acrodef{URLLC}[URLLC]{ultra-reliable low-latency communication}
\acrodef{SCL}[SCL]{successive cancellation list}
\acrodef{uBLER}[uBLER]{undetected block error probability}
\acrodef{TO}[TO]{transmission occasion}
\acrodef{LLR}[LLR]{log-likelihood ratio}
\acrodef{BLER}[BLER]{block error rate}
\acrodef{uBLER}[uBLER]{undetected block error rate}
\acrodef{IR}{Incremental redundancy}
\acrodef{IR-HARQ}[IR-HARQ]{\ac{IR}-hybrid automatic repeat request}
\acrodef{HARQ}[HARQ]{hybrid automatic repeat request}
\acrodef{CRC}[CRC]{cyclic redundancy check}
\acrodef{RTT}[RTT]{round trip time}
\acrodef{SNR}[SNR]{signal-to-noise ratio}
\acrodef{ACK}[ACK]{positive acknowledgment}
\acrodef{NACK}[NACK]{negative acknowledgment}
\acrodef{LUT}[LUT]{look-up table}
\acrodef{ML}[ML]{machine learning}
\acrodef{MAP}[MAP]{maximum a posteriori}
\acrodef{DMRS}[DMRS]{demodulation reference signal}
\acrodef{MCS}[MCS]{modulation and coding scheme}
\acrodef{RV}[RV]{redundancy version}
\acrodef{LDPC}[LDPC]{low-density parity-check}
\acrodef{BP}[BP]{belief propagation}
\acrodef{SC}[SC]{successive cancellation}
\acrodefplural{i.i.d.}[i.i.d.]{independent and identically distributed}
\acrodef{NR}[NR]{new radio}
\acrodef{BG2}[BG2]{Base Graph 2}
\acrodef{QPSK}[QPSK]{quadrature phase-shift keying}
\acrodef{AWGN}[AWGN]{additive white Gaussian noise}
\acrodef{RV}[RV]{redundancy version}
\acrodef{SPA}[SPA]{sum product algorithm}

\begin{document}
\title{Enhanced Feedback Mechanisms for Resource-Efficient Incremental Redundancy}

\author{\IEEEauthorblockN{Mustafa Cemil Co\c{s}kun\IEEEauthorrefmark{1}, Ahmed Elkelesh\IEEEauthorrefmark{1}, Avijit Mandal\IEEEauthorrefmark{2}, and Homa Esfahanizadeh\IEEEauthorrefmark{1}}\\
\IEEEauthorblockA{\IEEEauthorrefmark{1}\{mustafa.coskun, ahmed.elkelesh, homa.esfahanizadeh\}@nokia-bell-labs.com; \IEEEauthorrefmark{2}avijit.mandal@duke.edu\\
\IEEEauthorrefmark{1}Nokia Bell Labs; \IEEEauthorrefmark{2}Duke University}}
\maketitle

\begin{abstract}
Incremental redundancy (IR) can reduce error rates by spreading coded bits across multiple transmission attempts. However, conventional stop-and-wait operation with coarse feedback often over-provisions retransmissions, triggers unnecessary decoding attempts, and increases end-to-end latency. This paper develops enhanced feedback and scheduling mechanisms that predict the additional redundancy needed for successful decoding and allocate only the required resources. We study two complementary strategies. First, using channel statistics, we learn a one- or two-shot mapping from channel quality to the minimum redundancy budget. As a byproduct, we derive an achievable reliability lower bound on the error probability of hybrid automatic repeat request (HARQ) systems. Numerical results with a polar-coded IR-HARQ scheme show that the bound can be closely approached by appropriately selecting the second-transmission redundancy over a wide SNR range with savings up to 60\% in retransmission size. Second, we propose a realization-aware early-feedback mechanism that uses first-transmission reliability information to make per-codeword decisions before decoding: whether the codeword is already decodable, if not, how many additional redundancy versions are needed, or whether decoding is unlikely and rate adaptation is preferable. Link-level simulations with 5G new radio LDPC codes show that both predictors achieve high accuracy (about 94.5\% in our study), increasing the probability of successful decoding within at most two transmission occasions.
\end{abstract}

\section{Introduction}%
\label{sec:intro}%

Emerging services over wireless networks require reliable and low-latency communication even under bandwidth constraints. For these delay-sensitive applications, conventional retransmission rounds can be costly, especially when feedback and scheduling decisions are separated in time. \ac{IR}, as envisioned in \ac{5G-NR}, enables additional resources to be transmitted in multiple stages based on receiver feedback and decoding status. However, existing incremental-redundancy mechanisms often rely on fixed feedback rules and conservative resource allocation, which can lead to unnecessary retransmissions and increased decoding delay. In this paper, we develop enriched feedback and scheduling mechanisms that predict the appropriate additional redundancy for successful decoding and allocate only the required resources, thereby reducing receiver complexity, improving throughput, and lowering end-to-end latency.

In \ac{5G-NR}, \ac{IR} is standardized for \ac{LDPC} codes via \acp{RV}. A mother codeword is stored in a circular buffer, from which an appropriate number of bits is read starting at specific bit locations determined by the \ac{RV} index, i.e., a mechanism known as rate-matching. Different \ac{RV} indices select different bit subsets for transmission, effectively dividing an \ac{LDPC} codeword into four possibly overlapping chunks. For a transmission configured with multiple \acp{TO}, the RV sequence is determined by the RV ID (see Table~5.1.2.1-2 in TS~38.214 \cite{3gpp38214}). The initial RV can be $0$, $1$, $2$, or $3$, and subsequent occasions usually follow the circular pattern $0\!\to \!2\!\to \!3\!\to\! 1$. The decoder at the receiver can attempt to decode using any combination of received RVs, where it feeds back \ac{ACK} or \ac{NACK} depending on the outcome. Some \acp{RV} are self-decodable, enabling low-latency decoding, while some others may not achieve this depending on the rate\cite{Sarkis21:IR-HARQ}. This standardized mechanism, which is also illustrated as Fig. \ref{fig:nr_chain_feedback}, motivates the need for smarter feedback and scheduling that can avoid sending unnecessary redundancy versions while still meeting stringent latency and reliability targets.
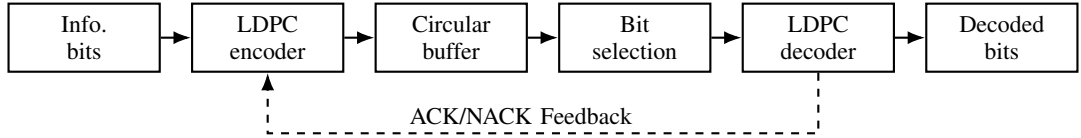
\begin{figure*}[t]
\centering
\begin{tikzpicture}[
    >=Latex,
    font=\small,
    blk/.style={
        draw,
        thick,
        minimum width=20mm,
        minimum height=9mm,
        align=center
    },
    arr/.style={->, thick}
]
\node[blk] (info) {Info.\\bits};
\node[blk, right=4mm of info] (enc) {LDPC\\encoder};
\node[blk, right=4mm of enc] (cb) {Circular\\buffer};
\node[blk, right=4mm of cb] (sel) {Bit\\selection};
\node[blk, right=4mm of sel] (dec) {LDPC\\decoder};
\node[blk, right=4mm of dec] (out) {Decoded\\bits};

\draw[arr] (info) -- (enc);
\draw[arr] (enc) -- (cb);
\draw[arr] (cb) -- (sel);
\draw[arr] (sel) -- (dec);
\draw[arr] (dec) -- (out);

\draw[arr, dashed] (dec.south) |- ++(0,-8mm)
  -| node[pos=0.27, above=2pt, fill=white, inner sep=1pt] {ACK/NACK Feedback} (enc.south);
\end{tikzpicture}
\caption{Simplified NR coding chain with HARQ feedback.}
\label{fig:nr_chain_feedback}
\end{figure*}

Beyond LDPC codes, polar-coded \ac{IR-HARQ} schemes (not yet standardized, but there are significant efforts to improve their performance to this end, e.g., see \cite{zhang20261tondecodingcapacityawarerateless,Gross26:IR-HARQ-FastSC} and references therein) also face limitations in optimizing retransmission strategies. Most existing \ac{IR-HARQ} works, e.g., \cite{wesel18, Li16:rateless_polar, MA17:IRHARQpolar, Gross25:IR-HARQ-FastSCL}, assume that any block error can be detected, which holds approximately only if a large outer \ac{CRC} code is used. This incurs significant rate loss for \ac{URLLC} applications with short messages. In this work, we formally present the fact that the \ac{uBLER} of the initial transmission establishes a lower bound on the overall \ac{BLER} of any \ac{IR-HARQ} system. Ignoring this constraint makes optimal bit allocation for retransmissions infeasible, leading to suboptimal resource usage and unnecessary latency. We leverage this observation and propose allocating just enough resources for a single retransmission that matches the reliability lower bound based on the received \ac{SNR}. Since polar codes\cite{arikan2009} are more susceptible to undetectable errors than \ac{5G-NR} LDPC codes in practice,\footnote{This is due to the fact that \ac{SCL} decoding is a complete decoder if CRC is skipped (or approximately complete if CRC is short in practice to achieve competitive \ac{BLER} performance for short messages).} the simulations are provided for a polar-coded \ac{IR-HARQ} scheme.

Machine learning techniques have been used to predict \textit{in advance} whether a received packet is decodable, i.e., whether the decoder will converge, before any expensive decoding attempt \cite{decodability-patent, uniBremen}. Along the same lines, learned models have been employed to pre-select which decoder can successfully perform the decoding task while minimizing cost \cite{SiPS-2020}, to estimate how many \ac{BP} iterations will be needed for convergence \cite{decodability-patent}, and to predict the minimum list size needed for successful \ac{SCL} decoding of polar codes \cite{decodability-patent}. However, these approaches either provide a binary decision (i.e., decodable/non-decodable) or target decoder-internal parameters, but none of them estimate how much additional redundancy is needed for successful decoding, which is the main focus of this work. Along similar lines, the authors of \cite{petar25:IR-harq-learned} propose to learn which symbols to retransmit in the case of a failure, which improves the network performance significantly. Nevertheless, such feedback to signal which coded symbols to retransmit becomes increasingly expensive as the code length increases.

We propose enriched feedback and scheduling for \ac{IR} through two contributions. The first uses channel statistics (e.g., \ac{SNR}) to derive a fixed resource-allocation policy across codewords, while the second uses a trained machine-learning model on first-transmission \acp{LLR} to make per-codeword early-feedback decisions on additional redundancy. Specifically, the model predicts the minimum number of additional \acp{RV} required for the successful decoding of each individual codeword. This enables adaptive second-transmission scheduling that avoids both under-allocation (high failure probability) and over-allocation (wasted resources). Compared with fixed or binary-feedback strategies, the proposed methods improve throughput, reduce decoding complexity by avoiding unnecessary full decoding attempts, and lower end-to-end latency by increasing the probability of successful decoding within at most two transmissions.

The remainder of the paper is organized as follows. Section~\ref{sec:system_model} introduces the system model and performance metrics. Section~\ref{sec:contribution1} presents resource prediction based on channel statistics. Section~\ref{sec:contribution2} develops the realization-aware early-feedback mechanism based on first-transmission \acp{LLR}. Numerical results are provided in Section~\ref{sec:numerical_results}, and Section~\ref{sec:conc} concludes the paper.
\section{System Model and Preliminaries}
\label{sec:system_model}

We consider a point-to-point \ac{IR-HARQ} system where a transmitter encodes $k$ information bits into a mother codeword of length $n_{T_{\max}}$ using an $(n_{T_{\max}},k)$ channel code. The mother codeword is partitioned into $T_{\max}$ incremental blocks with lengths $\Delta n_1,\Delta n_2,\ldots,\Delta n_{T_{\max}}$. In the $j$-th \ac{TO}, the transmitter sends incremental block $j$, and the receiver combines all received blocks and attempts decoding based on the accumulated $n_j$ observations, where $n_j \coloneqq \sum_{i=1}^{j} \Delta n_i$ (equivalently, the first $n_j$ received coded bits), and $j\in\{1,\ldots,T_{\max}\}$.

Communication proceeds in a stop-and-wait fashion over a memoryless channel. In a conventional system, after each decoding attempt, the receiver feeds back a single bit: an \ac{ACK} (decoding successful) or a \ac{NACK} (requesting further redundancy). In this work, we generalize the feedback to a message $F_i \in \mathcal{F}$ that may convey richer information, e.g., an indicator for the predicted minimal number of additional redundancy bits $\Delta n_{i+1}$ required for successful decoding in the next round. The classical binary feedback is the special case $\mathcal{F} = \{\text{ACK}, \text{NACK}\}$. For the $i$-th \ac{TO}, $E_i$, $E_i^{(d)}$ and $E_i^{(u)}$ denote the block error event, the \emph{detected} error event, i.e., the event that an error is present and the receiver is aware of it (e.g., via a CRC check), and the \emph{undetected} error event, i.e., the event that an error is present but the receiver is not aware of it (e.g., when the CRC check is satisfied), respectively. The complement $\overline{E}_i$ denotes successful decoding for the $i$-th \ac{TO} and $E$ denotes the block error event after the receiver stops asking for retransmission due to either an undetected error or reaching the maximum number of allowed retransmissions. Any communication system aims at minimizing $\Pr\left\{E\right\}$ subject to latency and resource constraints.

\subsection{Timing Model and Latency}
\label{sec:latency}
We adopt a slot-based timing model in which each transmission--feedback cycle occupies one \ac{RTT}. One \ac{RTT} encompasses the transmission time, propagation delay, receiver processing (including a decoding attempt), and the return of feedback. This model is appropriate when the coded block lengths under consideration fit within a single scheduling slot. 

Let $T \in \{1, 2, \ldots, T_{\max}\}$ denote the random variable indicating the number of transmission rounds required to either successfully decode or exhaust all redundancy. The \emph{expected latency} per information block is
\begin{equation}
\label{eq:latency}
  L \;\coloneqq\; \mathbb{E}[T] \cdot \mathrm{RTT}
\end{equation}
where the average number of transmissions is given by
\begin{equation}\label{eq:avg_transmission}
    \mathbb{E}[T] = 1 + \sum_{i=1}^{T_{\max}-1} \Pr\!\left\{ E_{1:i}^{(d)} \right\}
\end{equation}
with the definition $E_{1:i}^{(d)}\coloneqq E_1^{(d)}\cap E_2^{(d)}\cap\cdots\cap E_i^{(d)}$ and the convention that the first round is mandatory. In a practically relevant two-transmission case, i.e., $T_{\max} = 2$, the expression \eqref{eq:avg_transmission} reduces to
\begin{equation}
\label{eq:latency_two}
  \mathbb{E}[T] \;=\; 1 + \Pr\!\left\{E_1^{(d)}\right\}.
\end{equation}
Note that, under this model, the latency depends on the detected error probability of the first-transmission but \emph{not} on the retransmission length $\Delta n_2$. This suggests to choose $\Delta n_2$ as large as possible while fitting within a single scheduling slot in order to minimize $\Pr\left\{E\right\}$; however, the retransmission length affects the resource consumption and throughput defined next.

\subsection{Throughput}
\label{sec:throughput}

Throughput, denoted as $\eta$, measures how efficiently the channel is used per successfully delivered information bit. It is defined as $\eta \coloneqq \frac{\mathbb{E}[K]}{\mathbb{E}[N]}$ where $\mathbb{E}[K]$ is the expected number of successfully delivered information bits per block and $\mathbb{E}[N]$ is the expected total number of transmitted coded bits per message. Hence, it should be computed as
\begin{equation}
    \eta = \frac{\mathbb{E}[K]}{\mathbb{E}[N]} =   \frac{k \sum_{i=1}^{T_{\max}} \Pr\!\left\{E_{1:i-1}^{(d)}\cap \overline{E}_i \right\}}{n_1 + \sum_{i=2}^{T_{\max}} \Delta n_i\,\Pr\!\left\{E_{1:i-1}^{(d)}\right\}} \label{eq:tput}
\end{equation}
with the convention that empty intersections equal the certain event. For the case of $T_{\max} = 2$, throughput \eqref{eq:tput} reduces to
\begin{equation}
    \eta = \frac{k\,\left(\Pr\!\left\{\overline{E}_1\right\}
              \;+\; \Pr\!\left\{\overline{E}_2 \cap E_1^{(d)}\right\}\right)}{n_1 + \Delta n_2\,\Pr\!\left\{E_1^{(d)}\right\}}.
\end{equation}
Observe that increasing $\Delta n_2$ results in larger $\Pr\left\{\overline{E}_2 \cap E_1^{(d)}\right\}$; hence, increased $\mathbb{E}[K]$. Nevertheless, too many redundancy bits in a retransmission reduces $\eta$ by increasing $\mathbb{E}[N]$ while allocating too few causes decoding failure resulting in a non-optimal $\Pr\left\{E\right\}$. Hence, the goal of this work is to predict the minimum $\Delta n_2$ that ensures maximum reliability under strict latency requirements, e.g., $T_{\max}=2$, without wasting unnecessary resources which may reduce throughput significantly.

\section{Resource Prediction from Channel Statistics}
\label{sec:contribution1}

This section develops a channel-statistics-driven feedback mechanism for one-/two-shot scheduling of \ac{IR}. Using an \ac{SNR} estimate and the selected \ac{MCS}, the receiver predicts the minimum redundancy budget needed to meet a target reliability and feeds back this compact decision, reducing unnecessary retransmissions and decoding attempts.

\subsection{One-Shot Resource Prediction for Reliable Decoding}
We use the channel statistics (e.g., \ac{SNR}) together with the selected \ac{MCS} to minimize the resource budget required to meet a target one-shot-transmission reliability. Concretely, we construct an SNR-based \ac{LUT} that indicates how many resources are needed for a target success rate within a single \ac{TO}. With this one-shot scheduling, all required redundancy is allocated upfront, minimizing the scheduling effort and retransmission overhead, and thereby increasing the overall throughput.

\subsection{Two-Shot Resource Prediction with \ac{uBLER} Constraints}
\label{sec:uBLERtheorem}
Consider an \ac{IR-HARQ} scheme that uses up to $T_{\max}$ transmission rounds with an error detection mechanism after each transmission. 
\begin{lemma}
\label{prop:floor}
For any \ac{IR-HARQ} scheme, the overall block error probability satisfies
\begin{equation}
\Pr\!\left\{E\right\} \geq \Pr\{E_1^{(u)}\}.
\label{eq:floor}
\end{equation}
\end{lemma}
\begin{IEEEproof}
Observe that any undetected error event after the first transmission terminates the retransmission events and leads to a block error; hence, we have
\begin{equation}
    E_1^{(u)} \subseteq E.
    \label{eq:subevent}
\end{equation}
This concludes the error floor behaviour described by \eqref{eq:floor}.
\end{IEEEproof}
Let $\Delta n\coloneqq \sum_{i=2}^{T_{\max}}\Delta n_i$. Next, we show that the equality in Equation \eqref{eq:floor} is achievable, which will lead to an effective IR-HARQ design principle.
\begin{theorem} \textbf{(Achievability of the lower bound)}
    The equality holds when $\Delta n\rightarrow \infty$ under an optimum decoder, e.g., \ac{MAP} decoder.
\end{theorem}
\begin{IEEEproof}
Using the law of total probability, we have
\begin{align}
\!\!\!\!\!\Pr\{E\} 
&= \Pr\{E\cap \overline{E_1}\} + \Pr\{E \cap E_1^{(u)}\} + \Pr\{E \cap E_1^{(d)}\} \label{eq:decomp1}\\
&= \Pr\{E_1^{(u)}\} + \Pr\{E_1^{(d)}\}\cdot \Pr\{E \mid E_1^{(d)}\}
\label{eq:decomp2}
\end{align}
where \eqref{eq:decomp2} follows from Equation \eqref{eq:subevent} together with the chain rule and by observing that $E$ and $\overline{E_1}$ are disjoint.

It remains to show that the conditional term $\Pr\{E \mid E_1^{(d)}\}$ vanishes asymptotically as $\Delta n\to\infty$ with $k$ fixed. Let $M\in\{1,\dots,2^k\}$ denote the (uniform) message index and let $Y^{n_{T_{\max}}}$ denote the collected channel output with cardinality $n_{T_{\max}}=\Delta n_1+\Delta n$. To bound the \ac{MAP} decoding error probability, note that the conditional probability of correct decoding satisfies
\begin{align}
    1-\Pr\{E \mid E_1^{(d)}\} &= \mathbb{E}\Big[\max_{m} P_{M|Y^{n_{T_{\max}}}}\big(m\,\big|\,Y^{n_{T_{\max}}}\big)\Big] \label{eq:MAP1}\\
    &\geq \mathbb{E}\Big[2^{-H\big(P_{M|Y^{n_{T_{\max}}}}(\cdot\mid Y^{n_{T_{\max}}})\big)}\Big] \label{eq:MAP2}\\
    &\geq 2^{-\mathbb{E}\Big[H\big(P_{M|Y^{n_{T_{\max}}}}(\cdot\mid Y^{n_{T_{\max}}})\big)\Big]} \label{eq:MAP3}\\
    &= 2^{-H\big(M\mid Y^{n_{T_{\max}}}\big)} \label{eq:MAP4}
\end{align}
where \eqref{eq:MAP1} is the definition of \ac{MAP} decoding (the expectation is over $Y^{n_{T_{\max}}}$), \eqref{eq:MAP2} uses the inequality $\max_i p_i\ge 2^{-H(\mathbf p)}$ for any probability distribution $\mathbf p=(p_1,\dots,p_{2^k})$, \eqref{eq:MAP3} follows from Jensen's inequality (since $2^{-x}$ is convex), and \eqref{eq:MAP4} is the definition of conditional entropy. It therefore suffices to show that $H\big(M\mid Y^{n_{T_{\max}}}\big)\to 0$ as $\Delta n$ grows, i.e., as $n_{T_{\max}}\to\infty$. 

Since $I\big(M;Y^{n}\big)$ is monotone non-decreasing in $n$ and upper bounded by $H(M)=k$, the limit $\lim_{n\to\infty} I(M;Y^n)$ exists. Next, we rule out convergence to a strict sub-ceiling $L<k$. For a non-degenerate channel and a reasonable IR-HARQ scheme, the conditional mutual information $I(M;Y_{n+1}\mid Y^n)$ is strictly positive whenever uncertainty about message $M$ remains. This means that $L=k$, and  we have
\begin{equation}
H\big(M\mid Y^{n_{T_{\max}}}\big)=k-I\big(M;Y^{n_{T_{\max}}}\big) \xrightarrow[\Delta n\to\infty]{} 0 \label{eq:limit}
\end{equation}
Combining \eqref{eq:limit} with \eqref{eq:MAP4} yields $\Pr\{E \mid E_1^{(d)}\}\to 0$. Substituting into \eqref{eq:decomp2} concludes the proof.
\end{IEEEproof}
Although it is an asymptotic result, Theorem 1 sets a realistic design principle for the second transmission as to minimize the error probability without spending more than needed bits for the second transmission. In other words, we choose the smallest $\Delta n_2$ such that $\Pr\left\{E\right\} \leq \beta\Pr\{E_1^{(u)}\}$ with $\beta\geq 1$. For the numerical examples, we choose $\beta$ to be very close to $1$ so that the optimization looks for the approximation $\Pr\left\{E\right\} \approx \Pr\{E_1^{(u)}\}$ to decide for $\Delta n_2$ for various \ac{SNR} values. In practice, \ac{LUT}s or expanded MCS tables can be constructed by fixing $\beta$ for chosen applications\cite{patent_polar_irharq}.

\section{Resource Prediction from First Transmission Occasion}
\label{sec:contribution2}

The mechanism proposed in this section makes \emph{per-codeword} scheduling decisions by exploiting reliability information available after the first \ac{TO} and before any decoding attempt is made. Specifically, a lightweight trained predictor maps first-transmission observations (e.g., the received \acp{LLR}) to the minimum additional redundancy needed for successful decoding, producing the enriched feedback message $F_1 \in \mathcal{F}$ introduced in Section~\ref{sec:system_model}.

As a concrete setup, consider \ac{5G-NR} \ac{LDPC} codes with up to four available \acp{RV} (RV~0 through RV~3), where each RV corresponds to one incremental resource unit. After the first \ac{TO}, the predictor outputs one of the following decisions for each codeword:
\begin{enumerate}
    \item the codeword is already decodable with the received first \ac{RV} alone (no retransmission needed);
    \item the number of additional \acp{RV} required for successful decoding (one, two, or three).
    \item the codeword cannot be decoded even with all available redundancy, prompting an \ac{MCS} adjustment rather than further retransmission.
\end{enumerate}

This early-feedback approach results in two major benefits over classical binary \ac{ACK}/\ac{NACK}:
\begin{itemize}
    \item \textbf{Resource efficiency:} Unnecessary retransmissions and redundant decoding attempts are avoided, since the predictor directly estimates the required resources without invoking the full decoder. Over-allocation is also prevented, freeing resources for other users.
    \item \textbf{Lower latency:} By scheduling the correct amount of redundancy in a single retransmission, multiple stop-and-wait rounds are eliminated. In the best case, successful decoding is achieved within at most two \acp{TO}, reducing the expected latency from $\mathbb{E}[T] \cdot \mathrm{RTT}$ under iterative HARQ to at most $2 \cdot \mathrm{RTT}$.
\end{itemize}

\section{Numerical Results}
\label{sec:numerical_results}
For the numerical results, we consider an \ac{AWGN} channel with \ac{QPSK} modulation, where \ac{SNR} is expressed as $E_s/N_0$, where $E_s$ is the energy per symbol and $N_0$ is the single-sided noise power spectral density. We consider two classes of channel codes considered in the \ac{5G-NR} standard, namely polar and \ac{LDPC} codes.
\subsection{Polar Codes}
In the case of polar codes, we picked the option with $11$-bit \ac{CRC} in the \ac{5G-NR} to obtain a trade-off between the undetected and overall \acp{BLER}. For the initial transmission, a $(256, 128)$ CRC-aided polar code is used, whose \ac{BLER} under \ac{SCL} decoding with list size $L=8$\cite{tal_vardy15} is provided in Fig. \ref{fig:bler}. As a reference, the \ac{BLER} of the $(512, 128)$ polar code is also provided in the same figure, where the codes are designed using the guidelines described in \cite{MA17:IRHARQpolar} such that the shorter one can be obtained from the longer one via puncturing.\footnote{Note that arbitrary lengths (not only powers of two) can be obtained using quasi-uniform puncturing\cite{Yuan18:IRHARQ} as we use in the simulations.} Observe that the \ac{BLER} curve for the \ac{IR-HARQ} scheme with $\Delta n_1=\Delta n_2 = 256$ matches that of $(512,128)$ code, as targeted by classical IR-HARQ designs \cite{Li16:rateless_polar,MA17:IRHARQpolar,Yuan18:IRHARQ}, until an error floor emerges due to the \ac{uBLER} performance of the first transmission as proven in Section \ref{sec:uBLERtheorem}. This means that, at lower SNR values than roughly $-1.5$ dB, improved reliability is possible without the need for a third transmission, and, at higher SNR values, fewer redundancy bits suffices to achieve the reliability lower bound. To this end, we obtain a vector as $\Delta n_2 (\mathbf{SNR}) = [544, 444, 256, 200, 144, 100]$ by optimizing $\Delta n_2$ as a function of $\mathbf{SNR} = [-3, -2, \dots, 2]$ via simulations. Then, the performance of the optimized \ac{IR-HARQ} scheme matches the target \ac{uBLER} tightly under the strict requirement of $T_{\max}=2$ transmissions, where the extra redundancy to be transmitted ranges from $544$ to $100$ bits.\footnote{Note also that, for total lengths longer than $512$ bits at low SNR values, we still use the method in \cite{MA17:IRHARQpolar} to achieve a high-performing code construction.} This means that, for the exemplary case of $\mathrm{SNR} = 2$ dB, $60\%$ less redundancy is transmitted compared to the scheme where the size of the retransmission is pre-allocated to be the same as the initial transmission. 
\begin{figure}[t]
\centering
\begin{tikzpicture}
\begin{semilogyaxis}[
    width=0.9\columnwidth,
    height=0.75\columnwidth,
    xlabel={SNR (dB)},
    ylabel={BLER/uBLER},
    xmin=-3.5, xmax=2,
    ymin=1e-4, ymax=2,
    grid=both,
    major grid style={dashed,gray!50},
    minor grid style={dashed,gray!25},
    legend style={
        font=\scriptsize,
        at={(0.02,-0.62)},
        anchor=south west,
        cells={anchor=west},
        fill=white,
        fill opacity=0.85,
        text opacity=1,
        draw=black!50,
    },
    legend cell align=left,
    tick label style={font=\small},
    label style={font=\small},
    mark size=2.2pt,
    every axis plot/.append style={line width=0.9pt},
    cycle list name=color list,
]

\addplot+[mark=o] coordinates {
    (-3.5, 1.000000) (-3.0, 1.000000) (-2.5, 1.000000)
    (-2.0, 0.999758) (-1.5, 0.997315) (-1.0, 0.987105)
    (-0.5, 0.946320) ( 0.0, 0.832610) ( 0.5, 0.619530)
    ( 1.0, 0.346480) ( 1.5, 0.132620) ( 2.0, 0.032670)
};
\addlegendentry{$(256, 128)$}

\addplot+[mark=triangle*] coordinates {
    (-3.5, 0.68493 ) (-3.0, 0.42    ) (-2.5, 0.18382 )
    (-2.0, 0.054289) (-1.5, 0.007654) (-1.0, 0.00081883)
};
\addlegendentry{$(512, 128)$}

\addplot+[mark=star] coordinates {
    (-3.5, 0.69444 ) (-3.0, 0.43431 ) (-2.5, 0.170068)
    (-2.0, 0.048321) (-1.5, 0.012204) (-1.0, 0.004693)
    (-0.5, 0.003760) ( 0.0, 0.003310) ( 0.5, 0.002320)
    ( 1.0, 0.001380) ( 1.5, 0.000440) ( 2.0, 0.000220)
};
\addlegendentry{IR-HARQ with $(\Delta n_1=\Delta n_2=256, k=128)$}

\addplot+[mark=*, dashed, color=cyan] coordinates {
    (-3.5, 0.005208) (-3.0, 0.004343) (-2.5, 0.005102)
    (-2.0, 0.003503) (-1.5, 0.003570) (-1.0, 0.003837)
    (-0.5, 0.003710) ( 0.0, 0.003310) ( 0.5, 0.002320)
    ( 1.0, 0.001380) ( 1.5, 0.000440) ( 2.0, 0.000220)
};
\addlegendentry{uBLER with $(256,128)$}

\addplot+[mark=+] coordinates {
    (-3.0, 0.0044478 ) (-2.0, 0.0041578 ) (-1.0, 0.0042651 )
    ( 0.0, 0.0033569 ) ( 1.0, 0.0013605 ) ( 2.0, 0.00015766)
};
\addlegendentry{IR-HARQ with $(\Delta n_1 = 256, \Delta n_2(\mathrm{SNR}),\,k)$}

\end{semilogyaxis}
\end{tikzpicture}
\caption{\ac{BLER}/\ac{uBLER} vs. \ac{SNR} on the \ac{AWGN} channel, where the proposed IR-HARQ scheme with varying $\Delta n_2$ per SNR is compared to the baseline with the same size for both first and the second transmission.}
\label{fig:bler}
\end{figure}
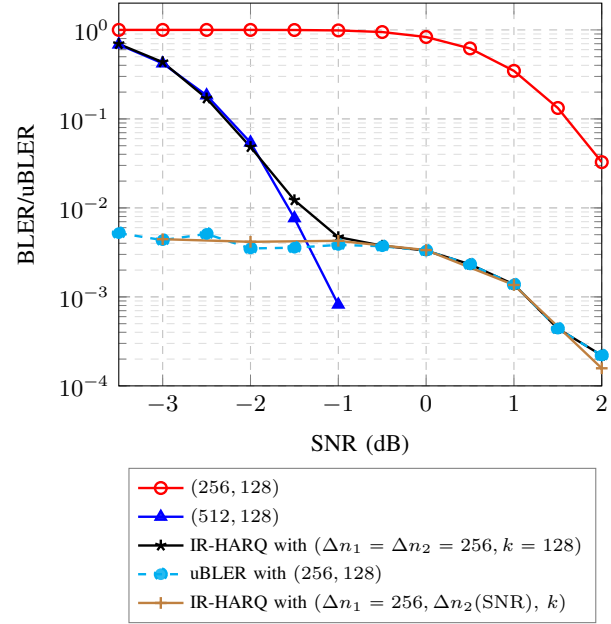
\subsection{LDPC codes}
\label{sec:dataset1}

For the \ac{LDPC}-based \ac{HARQ} study, we generate a labeled dataset using a link-level \ac{5G-NR} simulation chain implemented with the 5G Toolbox in MATLAB. Each packet is processed by \ac{CRC} attachment, \ac{5G-NR} \ac{LDPC} encoding, rate matching, \ac{QPSK} modulation, \ac{AWGN} transmission, soft demodulation, rate recovery, iterative \ac{LDPC} decoding (12 iterations of the flooding sum-product algorithm), and \ac{CRC} verification, with up to four \ac{HARQ} rounds. The transport-block size is fixed to $\mathrm{TBS}=3\,816$~bits with effective code rate $R = 308/1024$, and the \ac{RV} sequence is set to $[0,2,3,1]$. For each trial, $E_b/N_0$ is drawn uniformly from the set $\{-9,-8,\ldots,9\}$~dB. The input feature is the recovered first-round \ac{LLR} vector of length 19\,200 (corresponding to the maximum \ac{BG2} circular-buffer length). The label is the minimum number of \acp{RV} required for successful decoding: rounds~1--4 are mapped to classes~1--4, and failure after round~4 is mapped to class~5. To balance the classes, we use class-conditional acceptance (rejection sampling) to collect the same number of samples per class.

\subsubsection{Prediction Based on \ac{SNR}}
In practice, the receiver estimates the operating \ac{SNR} from reference signals such as the \ac{DMRS}. This estimate can serve as input to a predictor that determines, for a given \ac{MCS}, the minimum number of \acp{RV} required for successful decoding or whether successful decoding is infeasible with the available redundancy. We train a logistic-regression classifier that partitions the \ac{SNR} range into contiguous decision regions, each mapped to a specific resource allocation. A separate model is needed per \ac{MCS}.

Fig.~\ref{fig:snr_prediction_model_ldpc} illustrates the five-class prediction model for a representative \ac{5G-NR} \ac{LDPC} code and modulation from the generated dataset. The model assigns each code block to one of five classes. For precise channel-quality knowledge, the \ac{SNR}-based classifier achieves $96\%$ accuracy. However, when the channel estimate deviates from actual conditions, e.g., due to estimation error or mobility, the prediction accuracy degrades.

\begin{figure}[t]
    \centering
    \begin{subfigure}[t]{\columnwidth}
        \centering
        \includegraphics[width=0.8\columnwidth]{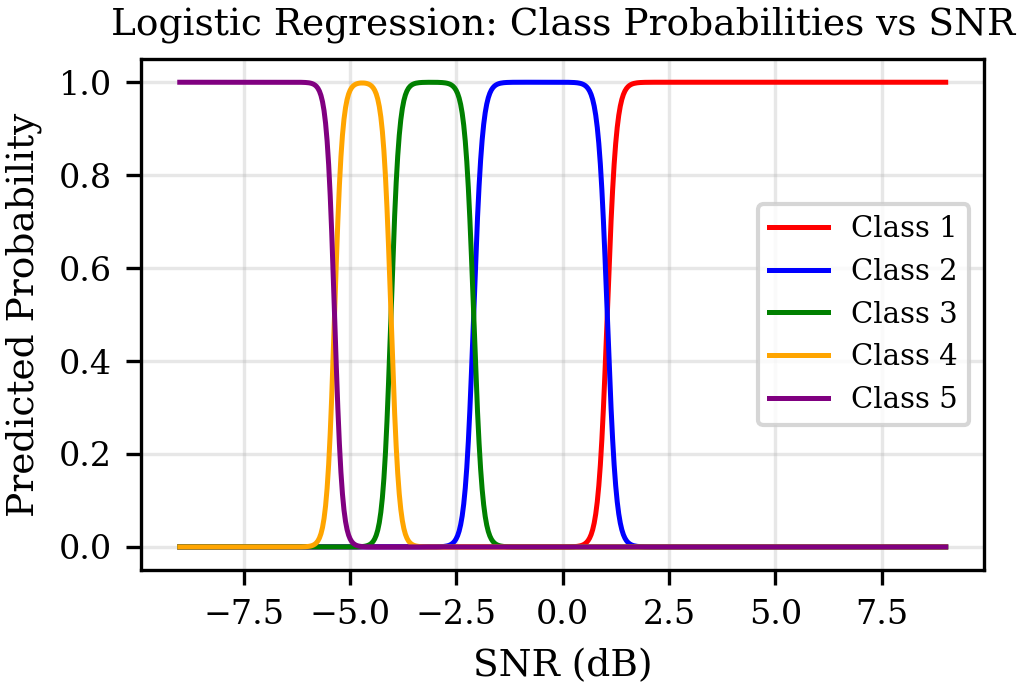}
    \end{subfigure}
    \begin{subfigure}[t]{\columnwidth}
        \centering
        \includegraphics[width=0.8\columnwidth]{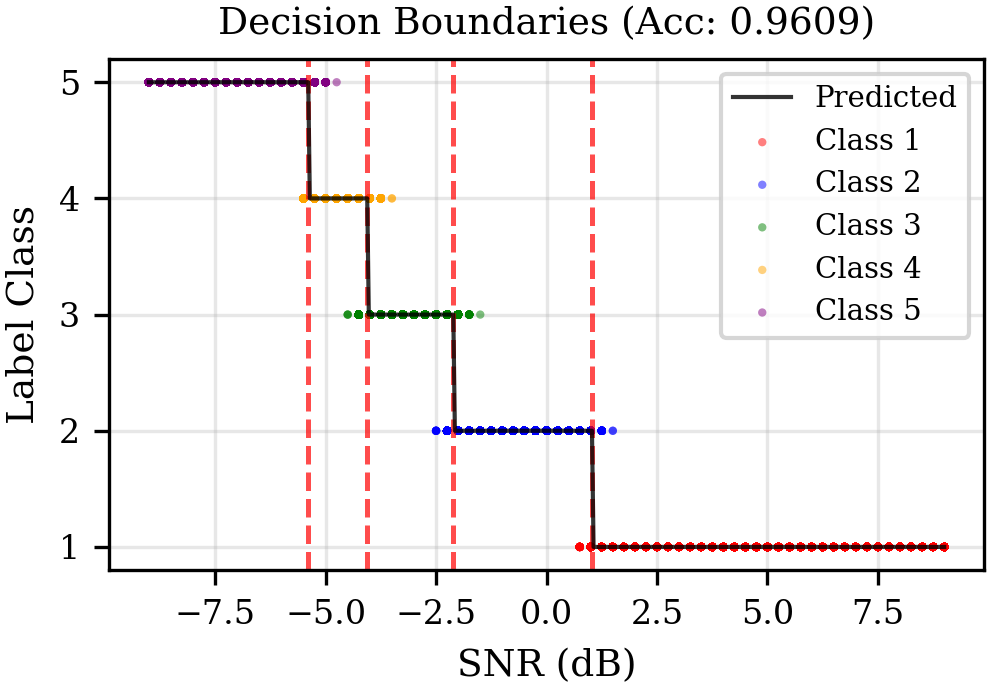}
    \end{subfigure}
    \caption{SNR-based resource prediction model for a representative 5G \ac{LDPC} code and modulation scheme. The colored points denote the true classes.}
    \label{fig:snr_prediction_model_ldpc}
\end{figure}

\subsubsection{Prediction Based on First-Transmission \acp{LLR}}

As an alternative to the \ac{SNR}-based approach, we exploit the \acp{LLR} of the first transmission directly to make per-codeword resource predictions, as described in Section~\ref{sec:contribution2}. This predictor bypasses the \ac{BP} decoder and outputs the five-class decision (decodable with the given resources, number of additional \acp{RV} needed, or undecodable).

The predictor is a one-dimensional CNN that treats the first-\ac{RV} \ac{LLR} vector as a sequence (Fig.~\ref{fig:llr_prediction_arch}): three convolutional blocks (Conv1d, batch normalization, ReLU, progressive pooling) followed by a fully connected head producing five-class logits. Inputs are clipped and $z$-score-normalized; training follows an $80\%$/$10\%$/$10\%$ split over the dataset. Optimization uses the Adam optimizer with cross-entropy loss, mini-batch training, and learning-rate scheduling. In the present configuration with four available redundancy versions, the network predicts one of five outcomes: the minimum number of \acp{RV} required for successful decoding, or that decoding is infeasible under the given \ac{MCS} and \ac{RV} order. If the predictor is extended to also suggest the \ac{RV} transmission order, additional output classes can be introduced.

\begin{figure}[t]
    \centering
    \includegraphics[width=0.91\columnwidth]{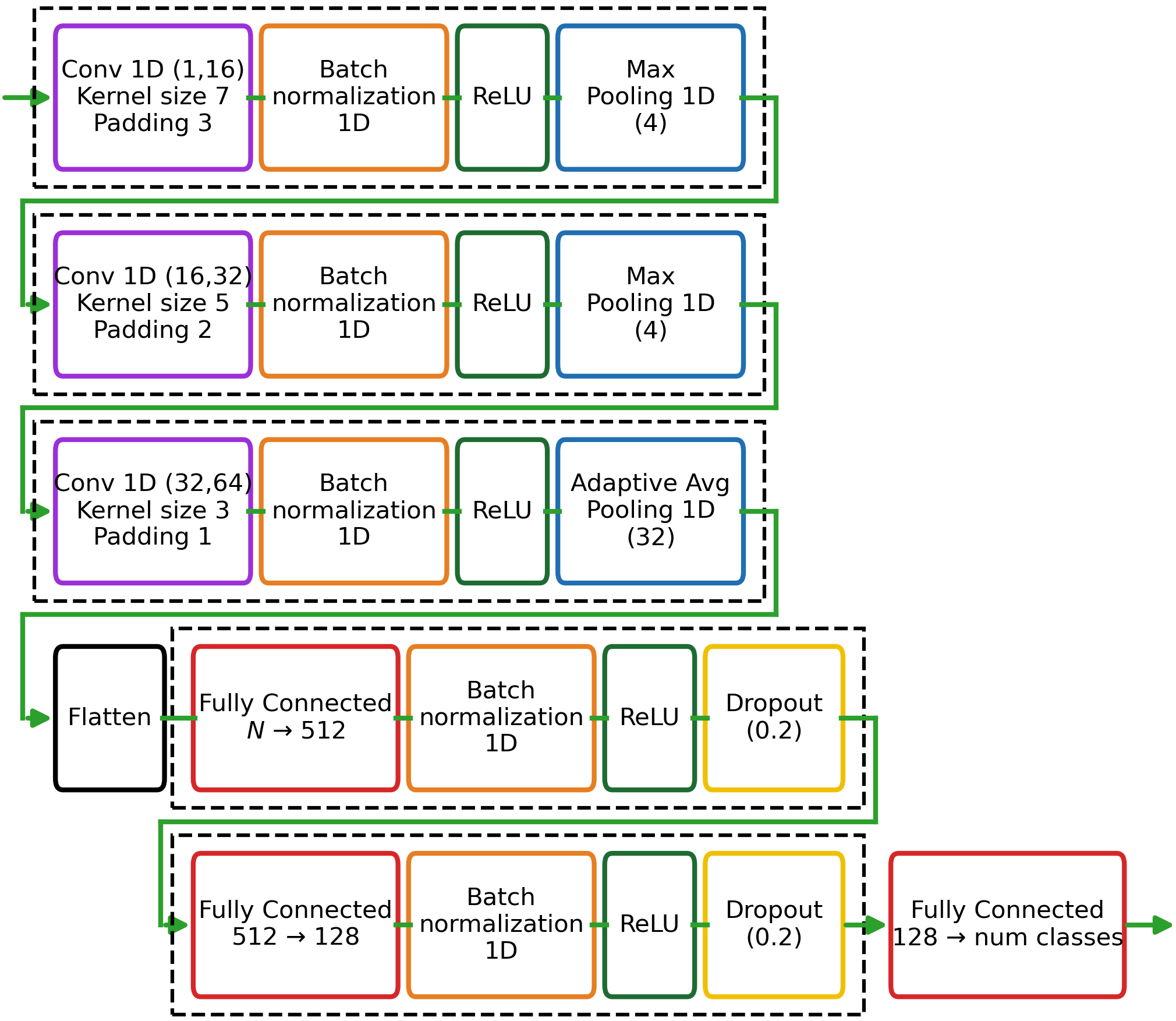}
    \caption{Architecture of the one-dimensional CNN predictor that maps first-transmission LLRs to resource allocation decisions.}
    \label{fig:llr_prediction_arch}
\end{figure}

The \ac{LLR}-based predictor achieves  $94.5\%$ accuracy on unseen test data. Fig.~\ref{fig:llr_confusion} shows the per-class accuracy and confusion matrix. The dominant error mode is over-allocation by one \ac{RV}: the network predicts one additional redundancy version beyond what is strictly required. This type of error does not degrade decoding performance but marginally reduces resource efficiency. For Class~5 misclassifications, where the model predicts that all four \acp{RV} are needed but decoding is, in fact, infeasible. Similarly, the error is not critical since the \ac{BP} decoder will detect the failure after combining and trigger an \ac{MCS} adjustment.

\begin{figure}[t]
    \centering
    \begin{tabular}{cc}
        \hspace{-0.2cm}\includegraphics[width=0.48\columnwidth]{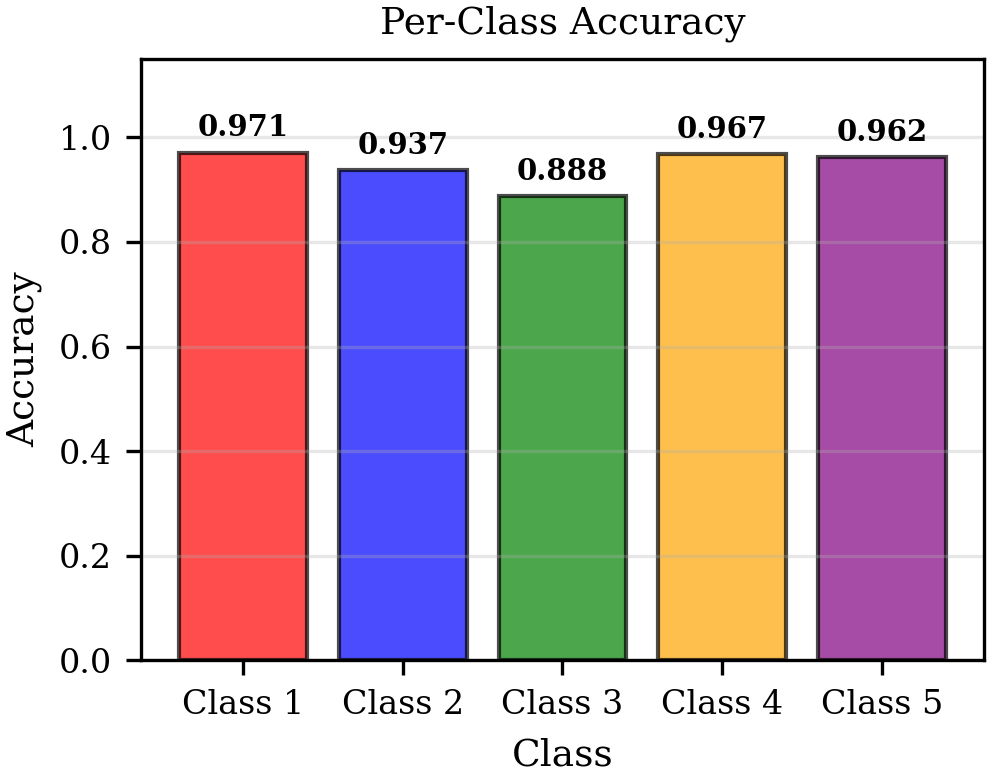} &\hspace{-0.2cm}
        \includegraphics[width=0.45\columnwidth]{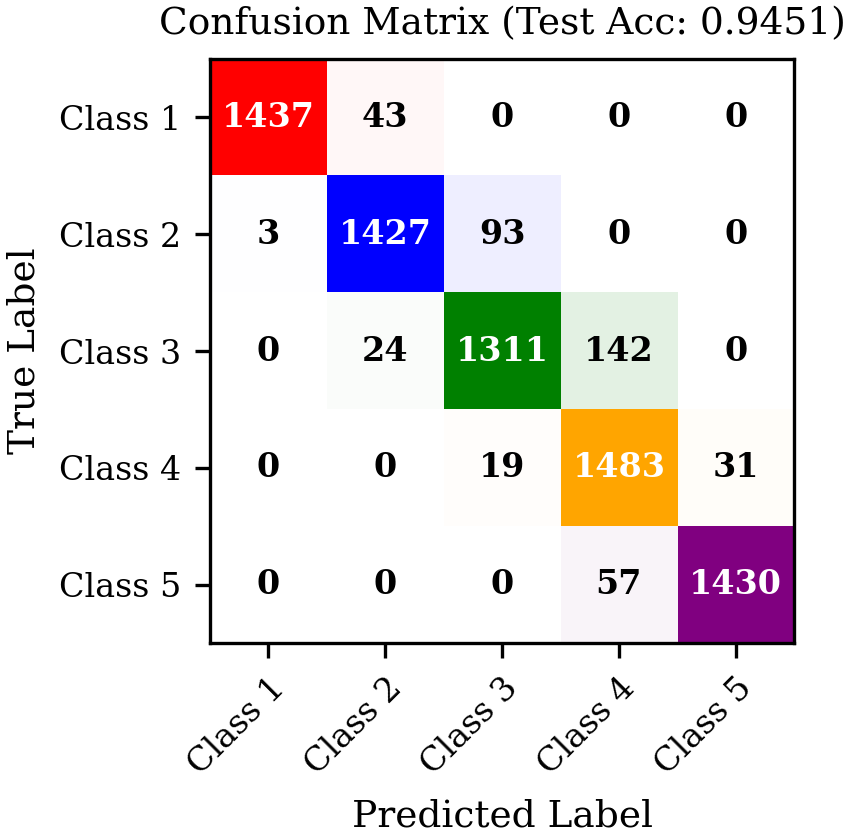}\\
        (a) & (b) 
    \end{tabular}
    \caption{Performance of the LLR-based predictor: (a) Per-class accuracy, showing the percentage of correct predictions for each of the five classes. (b) Confusion matrix, illustrating the distribution of true vs. predicted classes and highlighting the dominant error mode.}
    \label{fig:llr_confusion}
\end{figure}

The training dynamics of the classifier are illustrated in Fig.~\ref{fig:training_profile}, which shows that the model converges rapidly. The validation accuracy initially exceeds the training accuracy in the early epochs, because of dropout regularization (rate $= 0.3$), active only during training, which effectively reduces the model's capacity. At evaluation time, the full network is utilized. The ReduceLROnPlateau scheduler progressively reduces the learning rate upon validation accuracy plateaus, helping refine the model weights and achieve a best validation accuracy of 94.45\%.

\begin{figure}[t]
    \centering
    \includegraphics[width=0.99\linewidth]{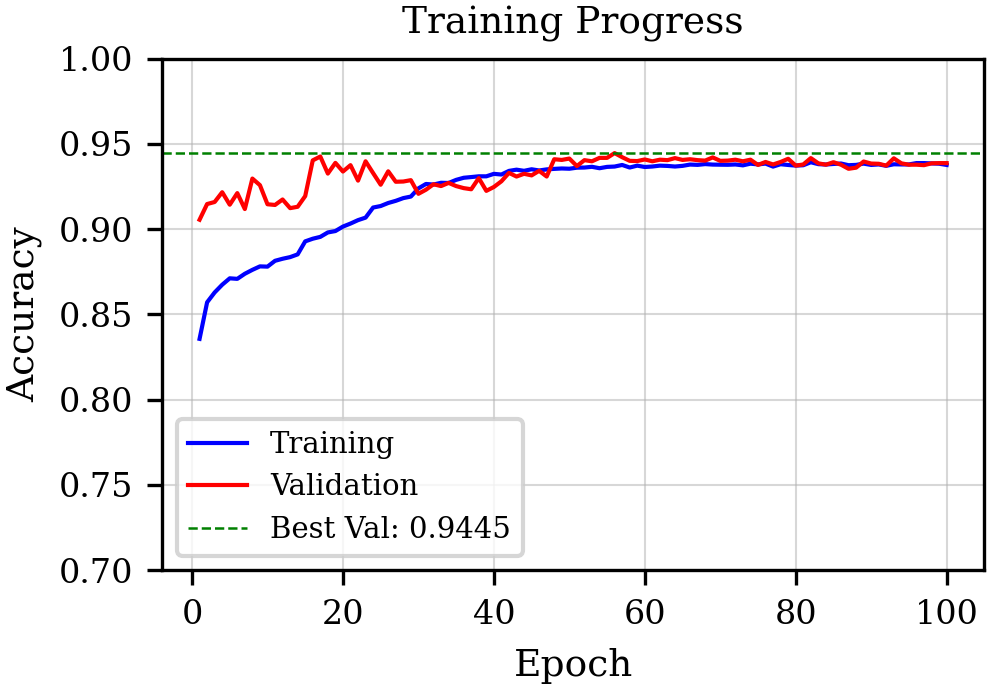}
    \caption{Training and validation accuracy of the classifier.}
    \label{fig:training_profile}
\end{figure}

\section{Conclusions}

\label{sec:conc}
We proposed enhanced feedback and scheduling schemes for incremental-redundancy hybrid automatic repeat request (IR-HARQ) that predict the required retransmission resources, considering IR-HARQ schemes both with 5G new radio (5G-NR) low-density parity-check (LDPC) codes and polar codes for future systems. We considered both a signal-to-noise ratio (SNR)-based policy and an early-feedback predictor based on first-transmission log-likelihood ratios. Link-level results with polar codes show that the reliability can be maximized with strict latency requirements and improved efficiency depending on the SNR. For the case of \ac{5G-NR} \ac{LDPC} codes, simulations show that accurate prediction of the required number of additional redundancy versions improves resource efficiency and reduces latency. 


\newpage

\end{document}